\documentclass[a4paper,orivec,envcountsame]{llncs}
\usepackage{misc0}
 
\newcommand{\inputs}{\ensuremath{\Imf}\xspace} 
\newcommand{\axiomsset}{\ensuremath{\Amf}\xspace} 
\renewcommand{\input}{\ensuremath{\Imc}\xspace} 
\newcommand{\axioms}{\ensuremath{\Omc}\xspace} 
\newcommand{\valuation}{\ensuremath{\Vmc}\xspace} 
 
\newcommand{\lab}[1]{\ensuremath{{\sf lab}({#1})}\xspace} 
\newcommand{\cproperty}{\ensuremath{\Pmc}\xspace} 
\newcommand{\consequence}{\ensuremath{\alpha}\xspace}

\newcommand{\pinpointing}[1]{\ensuremath{\varphi_{#1}}\xspace} 
\newcommand{\algorithm}{\ensuremath{A}\xspace}

\newcommand{\pre}[1]{\ensuremath{{\sf pre}({#1})}\xspace} 
\newcommand{\axm}[1]{\ensuremath{{\sf ax}_{#1}}\xspace} 
\newcommand{\axp}[1]{\ensuremath{{\sf ax}'_{#1}}\xspace} 
\newcommand{\res}[1]{\ensuremath{{\sf res}({#1})}\xspace} 
\newcommand{\fm}[2]{\ensuremath{{\sf fm}({#1},{#2})}\xspace} 

\newcommand{\Texa}{\ensuremath{\Tmc_{\sf exa}}\xspace}

\newcommand{\elplus}{\ensuremath{\EL^+}\xspace}
\newcommand{\wto}{\ensuremath{\rightharpoonup}\xspace}

\begin{document}
\title{Consequence-Based Axiom Pinpointing}
\author{Ana Ozaki \and Rafael Pe\~naloza}
\institute{KRDB Research Centre, Free University of Bozen-Bolzano, Italy}
\maketitle{}

\begin{abstract}
Axiom pinpointing refers to the problem of finding the axioms in an ontology that are relevant
for understanding a given entailment or consequence. One approach for axiom pinpointing,
known as \emph{glass-box}, is to modify a classical decision procedure for the entailments
into a method that computes the solutions for the pinpointing problem.
Recently, consequence-based decision procedures have been proposed as a  
promising alternative for tableaux-based reasoners  
for standard ontology languages. 
In this work, we present a general framework to extend 
consequence-based algorithms with axiom pinpointing. 
\end{abstract}
 
\section{Introduction}

Ontologies are now widely used in various domains such as medicine \cite{rector1995terminology,
price2000snomed,ruch2008automatic},
biology \cite{Sidhu05proteinontology}, 
chemistry \cite{degtyarenko2008chebi},
geography \cite{mcmaster2004research,kuipers1996ontological} and many others~\cite{article}, to represent 
conceptual knowledge in a formal and easy to understand manner. 
It is a multi-task effort to construct and maintain such ontologies, 
often containing thousands of concepts. 
As these ontologies increase in size and complexity, it becomes more and more challenging 
for an ontology engineer to understand which parts of the ontology cause 
a certain consequence to be entailed. If, for example, this consequence is an error, the ontology engineer 
would want to understand its precise causes, and correct it with minimal disturbances to the rest of the 
ontology.

To support this task, a technique known as
\emph{axiom pinpointing} was
introduced in~\cite{Schlobach:2003:NRS:1630659.1630712}. 
The goal of axiom pinpointing is to identify the minimal sub-ontologies (w.r.t.\ set inclusion) 
that entail a given consequence; we call these sets \emph{MinAs}. 
There are two basic approaches to axiom pinpointing.
The \emph{black-box approach}~\cite{parsia-www05} uses repeated calls to an unmodified decision procedure
to find these MinAs. The \emph{glass-box approach}, on the other hand, modifies the 
decision algorithm to generate the MinAs during one execution. In reality, glass-box methods
do not explicitly compute the MinAs, but rather a compact representation of them known as 
the \emph{pinpointing formula}.
In this setting, each axiom of the ontology is labelled with a unique 
propositional symbol. The pinpointing formula is a (monotone) Boolean 
formula, satisfied exactly by those valuations which evaluate to true 
the labels of the axioms in the ontology which cause the entailment of 
the consequence. Thus, the formula points out to the user the relevant parts of the 
ontology for the entailment of a certain consequence, where disjunction means 
alternative use of the axioms and conjunction means that 
the axioms are jointly used. 

Axiom pinpointing can be used to enrich a decision procedure for 
entailment checking by further presenting to the user the 
axioms which cause a certain consequence. 
Since glass-box methods modify an existing decision procedure, they require a specification
of the decision method to be studied. Previously, general methods for extending tableaux-based
and automata-based decision procedures to axiom pinpointing have been studied in 
detail~\cite{DBLP:journals/jar/BaaderP10,DBLP:journals/logcom/BaaderP10}.  
Classically, automata-based decision procedures often exhibit optimal worst-case complexity, but
the most efficient reasoners for standard ontology languages are tableaux-based. 
When dealing with pinpointing extensions one observes a similar behaviour: the automata-based
axiom pinpointing approach preserves the complexity of the original method, while tableau-based axiom
pinpointing is not even guaranteed to terminate in general. However, the latter are more goal-directed 
and lead to a better run-time in practice.

A different kind of reasoning procedure that is gaining interest is known as the 
consequence-based method. In this setting, rules are applied to derive explicit consequences from
previously derived knowledge.
Consequence-based decision procedures often enjoy optimal worst-case complexity 
and, more recently, they have  been presented as a 
promising alternative for tableaux-based reasoners  
for standard ontology languages~\cite{DBLP:conf/dlog/CucalaGH17,Kaza09,horrocks-kr16,DBLP:journals/ai/SimancikMH14,SiKH-IJCAI11,DBLP:conf/csemws/WangH12,DBLP:conf/dlog/KazakovK14a}.  
Consequence-based algorithms have been previously described as simple variants of tableau 
algorithms~\cite{baader-ki07}, and as syntactic variants of automata-based methods~\cite{HuPe17}.
They share the positive complexity bounds of automata, and the goal-directed nature of tableaux.

In this work, we present a general approach to produce axiom pinpointing 
extensions of consequence-based algorithms. Our driving example and use case is 
the extension of the consequence-based algorithm for entailment checking 
for the prototypical ontology language \ALC~\cite{DBLP:conf/dlog/KazakovK14a}. 
We show that the pinpointing extension does not change the ExpTime 
complexity of the consequence-based algorithm for \ALC. 

\section{Preliminaries}

We briefly introduce the notions needed for this paper. We are interested in the problem of understanding the 
causes for a consequence to follow from an ontology. We consider   
an abstract notion of ontology and consequence relation. For the sake of clarity, however, we instantiate these
notions to the description logic \ALC.  

\subsection{Axiom Pinpointing}

To keep the discourse as general as possible, we consider an \emph{ontology language} to define a class
\Amf of \emph{axioms}. An \emph{ontology} is then a finite set of axioms; that is, a finite subset of \Amf. We 
denote the set of all ontologies as \Omf. 
A \emph{consequence property} (or \emph{c-property} for short) is a binary relation
$\Pmc\subseteq\Omf\times\Amf$ that relates ontologies to axioms. 
If $(\Omc,\alpha)\in\Pmc$, we say that $\alpha$ is a \emph{consequence} of \Omc or alternatively, that
\Omc \emph{entails} $\alpha$.

We are only interested in relations that are
monotonic in the sense that for any two ontologies $\Omc,\Omc'\in\Omf$ and axiom $\alpha\in\Amf$ such that
$\Omc\subseteq\Omc'$, if $(\Omc,\alpha)\in\Pmc$ then $(\Omc',\alpha)\in\Pmc$. In other words, adding more
axioms to an ontology will only increase the set of axioms that are entailed from it. For the rest of this paper 
whenever we speak about a c-property, we implicitly assume that it is monotonic in this sense.

Notice that our notions of ontology and consequence property differ from previous work. 
In~\cite{DBLP:journals/jar/BaaderP10,DBLP:journals/logcom/BaaderP10}, c-properties are defined using two
different types of statements and ontologies are allowed to require additional structural constraints. The former
difference is just syntactic and does not change the generality of our approach. In the latter case, our setting
becomes slightly less expressive, but at the benefit of simplifying the overall notation and explanation of 
our methods. As we notice at the end of this paper, our results can be easily extended to the more general setting 
from~\cite{DBLP:journals/jar/BaaderP10,DBLP:journals/logcom/BaaderP10}.

When dealing with ontology languages, one is usually interested in deciding whether an ontology \Omc entails an 
axiom $\alpha$; that is, whether $(\Omc,\alpha)\in\Pmc$. In axiom pinpointing, we are more interested in the
more detailed question of \emph{why} it is a consequence. More precisely, we want to find the minimal 
(w.r.t.\ set inclusion) sub-ontologies $\Omc'\subseteq\Omc$ such that $(\Omc',\alpha)\in\Pmc$ still holds. These
subsets are known as \emph{MinAs}~\cite{DBLP:journals/jar/BaaderP10,DBLP:journals/logcom/BaaderP10},
\emph{justifications}~\cite{parsia-iswc07}, or 
\emph{MUPS}~\cite{Schlobach:2003:NRS:1630659.1630712}---among many other 
names---in the literature. Rather than enumerating all these sub-ontologies explicitly, one approach is to 
compute a formula, known as the pinpointing formula, that encodes them.

Formally, suppose that every axiom $\alpha\in\Amf$ is associated with a unique propositional variable
$\lab{\alpha}$, and let $\lab\Omc$ be the set of all the propositional variables corresponding to axioms
in the ontology \Omc. 
A \emph{monotone Boolean formula} $\phi$ over ${\sf lab}(\axioms)$ 
is a Boolean formula using only variables in  ${\sf lab}(\axioms)$ and the 
connectives for  conjunction ($\wedge$) and disjunction  ($\vee$). The constants
$\top$ and $\bot$, always evaluated to true and false, respectively, are 
also monotone Boolean formulae. We identify a propositional 
valuation with the set of variables which are true in it. 
For a valuation \valuation and a set of axioms \axioms, the \emph{\valuation-projection of \axioms}
is the set $\axioms_\valuation:=\{\alpha\in\axioms\mid 
\lab{\alpha}\in\valuation\}$.
Given a c-property $\cproperty$ and an axiom $\alpha\in\Amf$,
a monotone Boolean formula $\phi$ over ${\sf lab}(\axioms)$
is called a \emph{pinpointing formula} for $\Gamma$ w.r.t $\cproperty$
if for every valuation $\valuation\subseteq \lab{\axioms}$:
\begin{align*}
(\axioms_\valuation,\alpha) \in\cproperty\text{ iff } \valuation \text{ satisfies } \phi. 
\end{align*}

\subsection{Description Logics} 

Description logics (DLs)~\cite{BCNMP03} are a family of knowledge representation formalisms that have been
successfully applied to represent the knowledge of many application domains, in particular from the life sciences~\cite{article}. 
We briefly introduce, as a prototypical example, \ALC, which is the smallest propositionally closed description 
logic.

Given two disjoint sets $N_C$ and $N_R$ of \emph{concept names} and \emph{role names}, respectively, 
\ALC \emph{concepts} are defined through the grammar rule:
\[
C ::= A \mid \neg C \mid C\sqcap C \mid \exists r.C, 
\]
where $A\in N_C$ and $r\in N_R$. A \emph{general concept inclusion} (GCI) is an expression of the form
$C\sqsubseteq D$, where $C,D$ are \ALC concepts. A \emph{TBox} is a finite set of GCIs.

The semantics of this logic is given in terms of \emph{interpretations} which are pairs of the form 
$\Imc=(\Delta^\Imc,\cdot^\Imc)$ where $\Delta^\Imc$ is a finite set called the \emph{domain}, and $\cdot^\Imc$
is the \emph{interpretation function} that maps every concept name $A\in N_C$ to a set 
$A^\Imc\subseteq\Delta^\Imc$ and every role name $r\in N_R$ to a binary relation 
$r^\Imc\subseteq\Delta^\Imc\times\Delta^\Imc$. The interpretation function is extended to arbitrary \ALC
concepts inductively as shown in Figure~\ref{fig:semantics}.
\begin{figure}[tb]
\begin{align*}
(\neg C)^\Imc := {} & \Delta^\Imc\setminus C^\Imc \\
(C\sqcap D)^\Imc := {} & C^\Imc\cap D^\Imc \\
(\exists r.C)^\Imc := {} & \{d\in\Delta^\Imc\mid \exists e\in C^\Imc.(d,e)\in r^\Imc\}
\end{align*}
\caption{Semantics of \ALC}
\label{fig:semantics}
\end{figure}
Following this semantics, we introduce the usual abbreviations $C\sqcup D:=\neg(\neg C\sqcap \neg D)$,
$\forall r.C:=\neg(\exists r.\neg C)$, $\bot:=A\sqcap \neg A$, and $\top:=\neg\bot$. That is, $\top$ stands for 
a (DL) tautology, and $\bot$ for a contradiction.
The interpretation \Imc \emph{satisfies} the GCI $C\sqsubseteq D$ iff $C^\Imc\subseteq D^\Imc$. It is a 
\emph{model} of the TBox \Tmc iff it satisfies all the GCIs in \Tmc.

One of the main reasoning problems in DLs is to decide \emph{subsumption} between two concepts $C,D$ 
w.r.t.\ a TBox \Tmc; that is, to verify that every model of the TBox \Tmc also satisfies the GCI $C\sqsubseteq D$.
If this is the case, we denote it as $\Tmc\models C\sqsubseteq D$. 
It is easy to see that the relation $\models$ defines a c-property over the class \Amf of axioms containing all
possible GCIs; in this case, an ontology is a TBox.

%
%


The following example instantiates the basic ideas presented in this section.

\begin{example}
\label{exa:cons}
Consider for example the \ALC TBox \Texa containing the axioms
\begin{align*}
A\sqsubseteq\exists r.A: {} & \axm1, &
\exists r.A\sqsubseteq B: {} & \axm2, \\ 
A\sqsubseteq\forall r.B: {} & \axm3, & 
A\sqcap B\sqsubseteq\bot: {} & \axm4,
\end{align*}
where $\axm{i}, 1\le i\le 4$ are the propositional variables labelling the axiom.
It is easy to see that $\Texa\models A\sqsubseteq\bot$, and there are two justifications for this fact; namely, the
TBoxes $\{\axm1,\axm2,\axm4\}$ and $\{\axm1,\axm3,\axm4\}$. From this, it follows that 
$\axm1\land\axm4\land(\axm2\lor\axm3)$ is a pinpointing formula for $A\sqsubseteq\bot$ w.r.t.\ 
\Texa.
\end{example}

\section{Consequence-based Algorithms}\label{sec:consequence-based-alg}

Abstracting from particularities, a \emph{consequence-based algorithm} works on a set
$\Amc$ of \emph{consequences}, which is expanded through rule applications. Algorithms of this kind have
two phases. The \emph{normalization} phase first transforms all the axioms in an ontology into a suitable normal 
form. 
The 
\emph{saturation} phase initializes the set $\Amc$ of \emph{derived consequences} with the normalized ontology 
and applies the rules to expand it. The set \Amc is often called a \emph{state}.
As mentioned, the initial
state $\Amc_0$ contains the normalization of the input ontology \Omc. A \emph{rule} is of 
the form $\Bmc_0\to \Bmc_1$, where $\Bmc_0,\Bmc_1$ are finite sets of consequences.
This rule is \emph{applicable} to the state \Amc if 
$\Bmc_0\subseteq \Amc$ and $\Bmc_1\not\subseteq\Amc$. Its \emph{application} 
extends $\Amc$ to $\Amc\cup\Bmc_1$. \Amc is \emph{saturated} if no rule is applicable to it. The method 
\emph{terminates} if \Amc is saturated after finitely many rule applications, independently of the rule
application order chosen. 
For the rest of this section and most of the following, we assume that the input ontology is already in this normal 
form, and focus only on the second phase.

Given a rule $R=\Bmc_0\to \Bmc_1$, we use $\pre{R}$ and $\res{R}$ 
to denote the sets $\Bmc_0$ of premises that trigger $R$ and $\Bmc_1$ of consequences
resulting of its applicability, respectively.
If the state $\Amc'$ is obtained from \Amc through the application of the rule $R$, we write 
$\Amc\to_R \Amc'$, and denote $\Amc\to\Amc'$ if the precise rule used is not relevant.

Consequence-based algorithms derive, in a single execution, several axioms that are entailed from the input
ontology. Obviously, in general they cannot generate \emph{all} possible entailed axioms, as such a set may be
infinite (e.g., in the case of \ALC). Thus, to define correctness, we need to specify for every ontology \Omc, a 
finite set $\delta(\Omc)$ of \emph{derivable consequences} of \Omc.

\begin{definition}[Correctness]
A consequence-based algorithm is \emph{correct} for the consequence property \Pmc if for every ontology
\Omc, the following two conditions hold: (i) it terminates, and 
(ii) if $\Omc\to^*\Amc$ and \Amc is saturated, then for every derivable consequence $\alpha\in\delta(\Omc)$ it 
follows that $(\Omc,\alpha)\in\Pmc$ iff $\alpha\in\Amc$.
\end{definition}
That is, the algorithm is correct for a property if it terminates and is sound and complete w.r.t.\ the finite set
of derivable consequences $\delta(\Omc)$.

Notice that the definition of correctness requires that the resulting set of consequences obtained from the 
application of the rules is always the same, independently of the order in which the rules are applied. In other 
words, if $\Omc\to^*\Amc$, $\Omc\to^*\Amc'$, and $\Amc,\Amc'$ are both saturated, then $\Amc=\Amc'$. This
is a fundamental property that will be helpful for showing correctness of the pinpointing extensions in the next
section.

A well-known example of a consequence-based algorithm is the \ALC reasoning method from~\cite{SiKH-IJCAI11}. 
To 
describe this algorithm we need some notation. A \emph{literal} is either a concept name or a negated concept 
name. Let $H,K$ denote (possibly empty) conjunctions of literals, and $M,N$ are (possibly empty) 
disjunctions of concept names. For simplicity, we treat these conjunctions and disjunctions as sets.
The normalization phase transforms all GCIs to be of the form: 
\[
\bigsqcap_{i=1}^n A_i\sqsubseteq \bigsqcup_{j=1}^mB_j, \qquad 
A \sqsubseteq \exists r.B, \qquad
A \sqsubseteq \forall r.B, \qquad
\exists r.A \sqsubseteq B.
\]
For a given \ALC TBox \Tmc, the set $\delta(\Tmc)$ of derivable consequences contains all GCIs 
of the form $H\sqsubseteq M$ and $H\sqsubseteq N\sqcup \exists r.K$.
The saturation phase initializes $\Amc$ to contain the axioms in the (normalized) TBox, and applies the rules from
Table~\ref{tab:rules} until a saturated state is found.
\begin{table}[tb]
\caption{\ALC consequence-based algorithm rules $\Bmc_0\to \Bmc_1$.}
\label{tab:rules}
\centering
\begin{tabular}{@{}ll@{\quad }l@{\, }|@{\, }l@{}}
\toprule
&$\Bmc_0$ &  & $\Bmc_1$ \\
\midrule
1: & $\emptyset$ && $H\sqcap A\sqsubseteq A$\\
2: & $H\sqcap \neg A \sqsubseteq N\sqcup A$ 	& 	& $H\sqcap\neg A\sqsubseteq N$ \\
3: & $H\sqsubseteq N_1\sqcup A_1,\ldots,H\sqsubseteq N_n\sqcup A_n$ & $A_1\sqcap\cdots \sqcap A_n\sqsubseteq N$ & $H\sqsubseteq \bigsqcup_{i=1}^n N_i\sqcup N$ \\
4: & $H\sqsubseteq N\sqcup A$ & $A\sqsubseteq\exists r.B$ & $H\sqsubseteq N\sqcup\exists r.B$ \\
5: & $H\sqsubseteq M\sqcup\exists r.K, K\sqsubseteq N\sqcup A$ & $\exists r.A\sqsubseteq B$ & $H\sqsubseteq M\sqcup B\sqcup\exists r.(K\sqcap \neg A)$ \\
6: & $H\sqsubseteq M\sqcup\exists r.K, K\sqsubseteq\bot$ &  & $H\sqsubseteq M$ \\
7: & $H\sqsubseteq M\sqcup\exists r.K, H\sqsubseteq N\sqcup A$ & $A\sqsubseteq \forall r.B$ & $H\sqsubseteq M\sqcup N\sqcup\exists r.(K\sqcap B)$ \\
\bottomrule
\end{tabular}
\end{table}
After termination, one can check that for every derivable consequence $C\sqsubseteq D$ it holds that
$\Tmc\models C\sqsubseteq D$ iff $C\sqsubseteq D\in \Amc$; that is, this algorithm is correct for the
property~\cite{SiKH-IJCAI11}.

\begin{example}
\label{exa:cba}
Recall the \ALC TBox \Texa from Example~\ref{exa:cons}. Notice that all axioms in this TBox are already
in normal form; hence the normalization step does not modify it. The consequence-based algorithm starts with 
$\Amc:=\Texa$ and applies the rules until saturation.
One possible execution of the algorithm is 
%
\begin{align*}
\Amc_0 \  \to_5 \ A\sqsubseteq B \  \to_3 \  A \sqsubseteq \bot 
	     \ \to_7 \ A\sqsubseteq \exists r.(A\sqcap B) \ \to^* \ldots ,
\end{align*}
%
where $\Amc_0$ contains \Texa and the result of adding all the tautologies generated by the application of
Rule $1$ over it (see Figure~\ref{fig:cba}).
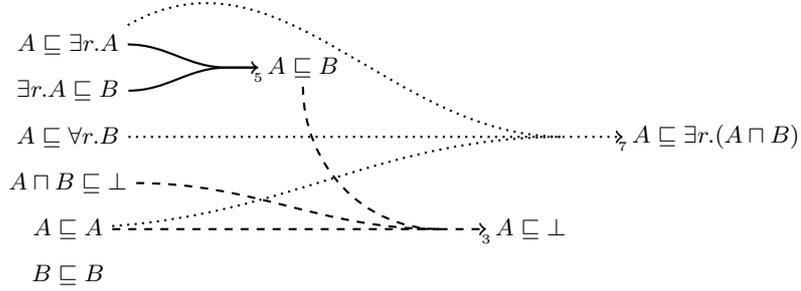
\begin{figure}[tb]
\centering
\begin{tikzpicture}[thick]
\node (a1) {$A\sqsubseteq\exists r.A$};
\node[below=1mm of a1] (a2) {$\exists r.A\sqsubseteq B$};
\node[below=1mm of a2] (a3) {$A\sqsubseteq\forall r.B$};
\node[below=1mm of a3] (a4) {$A\sqcap B\sqsubseteq \bot$};
\node[below=1mm of a4] (c1) {$A\sqsubseteq A$};
\node[below=1mm of c1] (c2) {$B\sqsubseteq B$};
\node[anchor=west] (c3) at ($(a1)!0.5!(a2)+(2.5,0)$) {$A\sqsubseteq B$};
\draw (a1) to[out=0,in=180] ($(c3)+(-1,0)$) -- (c3);
\draw (a2) to[out=0,in=180] ($(c3)+(-1,0)$) edge[->] node[below=-0.5mm,at end] {\tiny 5} (c3);
\node[anchor=west] (c4) at ($(a4)!1!(c1)+(5.5,0)$) {$A\sqsubseteq \bot$};
\draw[dashed] (a4) to[out=0,in=180] ($(c4)+(-1.2,0)$); 
\draw[dashed] (c3) to[out=-90,in=180] ($(c4)+(-1.2,0)$); 
\draw[dashed] (c1) to[out=0,in=180] ($(c4)+(-1.2,0)$) edge[->] node[below=-0.5mm,at end] {\tiny 3} (c4);
\node[right=6.5cm of a3,anchor=west] (c5) {$A\sqsubseteq \exists r.(A\sqcap B)$};
\draw[dotted] (a1.north east) to[out=30,in=180] ($(c5)+(-2,0)$); 
\draw[dotted,->] (a3) to node[below=-0.5mm,at end] {\tiny 7} (c5);
\draw[dotted] (c1) to[out=4,in=180] ($(c5)+(-2,0)$); 
\end{tikzpicture}
\caption{An execution of the \ALC consequence-based algorithm over \Texa from Example~\ref{exa:cba}.
Arrows point from the premises to the consequences generated by the application of the rule denoted in the
subindex.}
\label{fig:cba}
\end{figure}
Since rule applications only extend the set of consequences, we depict exclusively the newly added
consequence; e.g., the first rule application $\Amc_0\to_1 A\sqsubseteq A$ is in fact representing
$\Amc_0\to_1\Amc_0\cup\{A\sqsubseteq A\}$.
When the execution of the method terminates, the set of consequences \Amc contains $A\sqsubseteq\bot$;
hence we can conclude that this subsumption follows from \Texa. Notice that other consequences (e.g.,
$A\sqsubseteq B$) are also derived from the same execution.
\end{example}

For the rest of this paper, we consider an
arbitrary, but fixed, consequence-based algorithm, that is correct for a given c-property \Pmc.

%

\section{The Pinpointing Extension}

Our goal is to extend consequence-based algorithms from the previous section to methods that compute 
pinpointing 
formulae for their consequences. We achieve this by modifying the notion of states, and the rule applications
on them.
Recall that every axiom in $\alpha$ in the class \Amf (in hence, also every axiom in the ontology \Omc) is labelled 
with a unique propositional variable 
$\lab{\alpha}$. In a similar manner, we consider sets of consequences \Amc that are labelled with 
a monotone Boolean formula. 
We use the notation $\Amc^{pin}$ to indicate that the elements in a the set \Amc are labelled in this way,
and use $\alpha:\varphi_\alpha\in\Amc^{pin}$ to express that 
the consequence $\alpha$, labelled with the formula $\varphi_\alpha$, belongs to $\Amc^{pin}$.
A \emph{pinpointing state} is a set of labelled consequences. 
We assume that each consequence in this set is labelled with only one formula. 
For a set of labelled consequences $\Amc^{pin}$ and a set of (unlabelled) consequences $X$, we define
$\fm{X}{\Amc^{pin}}:=\bigwedge_{\alpha\in X}\varphi_\alpha$, where $\varphi_\alpha=\bot$ if 
$\alpha\notin\Amc$. 

A consequence-based algorithm $\algorithm$ induces a pinpointing consequence-based algorithm 
 $\algorithm^{pin}$ by modifying the notion of rule application, and dealing with pinpointing states, instead of 
 classical states, through a modification of the formulae labelling the derived consequences.
 
\begin{definition}[Pinpointing Application]
The rule $R=\Bmc_0\to \Bmc_1$ is \emph{pinpointing applicable} to the pinpointing state $\Amc^{pin}$ if 
$\fm{\Bmc_0}{\Amc^{pin}}\not\models\fm{\Bmc_1}{\Amc^{pin}}$.
The \emph{pinpointing application} of this rule
modifies $\Amc^{pin}$ to: 
\begin{align*}
\{\alpha:\varphi_\alpha\lor\fm{\Bmc_0}{\Amc^{pin}}\mid \alpha\in\Bmc_1,\alpha:\varphi_\alpha\in\Amc^{pin}\}
	\cup(\Amc^{pin}\setminus \{\alpha:\varphi_\alpha\mid \alpha\in \Bmc_1\}).
\end{align*}
The pinpointing state $\Amc^{pin}$ is \emph{pinpointing saturated} if no rule is pinpointing applicable to it.
\end{definition}
We denote as $\Amc^{pin}\wto_R\Bmc^{pin}$ the fact that $\Bmc^{pin}$ is obtained from the pinpointing
application of the rule $R$ to $\Amc^{pin}$. As before, we drop the subscript $R$ if the name of the rule is 
irrelevant and write simply $\Amc^{pin}\wto\Bmc^{pin}$. The pinpointing extension starts, as the classical one, with 
the set of all normalized axioms. For the rest of this section, we assume that the input ontology is already 
normalized, and hence each axiom in the initial pinpointing state is labelled with its corresponding propositional 
variable. In the next section we show how to deal with normalization.

\begin{example}
\label{exa:ppa}
Consider again the TBox \Texa from Example~\ref{exa:cons}. At the beginning of the execution of the pinpointing
algorithm, the set of consequences is the TBox, with each axiom labelled by the unique propositional variable
representing it; that is 
$\Tmc^{pin}=\{
 A\sqsubseteq\exists r.A:\axm1,
 \exists r.A\sqsubseteq B:\axm2,
 A\sqsubseteq\forall r.B:\axm3,
 A\sqcap B\sqsubseteq\bot:\axm4
\}$.
A pinpointing application of Rule 1 adds the new consequence $A\sqsubseteq A:\top$, where the tautology $\top$
labelling this consequence arises from the fact that rule 1 has no premises. At this point, one can pinpointing apply
Rule 5 with 
\[
\Bmc_0=\{A\sqsubseteq \exists r.A,\ A\sqsubseteq A,\ \exists r.A \sqsubseteq B\}, \qquad
\Bmc_1=\{A\sqsubseteq B\}
\] 
(see the solid arrow in Figure~\ref{fig:cba}). In this case,
$\fm{\Bmc_0}{\Amc^{pin}}=\axm1\land\top\land\axm2$ and $\fm{\Bmc_1}{\Amc^{pin}}=\bot$ because the 
consequence $A\sqsubseteq B$ does not belong to $\Amc^{pin}$ yet. Hence 
$\fm{\Bmc_0}{\Amc^{pin}}\not\models\fm{\Bmc_1}{\Amc^{pin}}$, and the rule is indeed pinpointing applicable. 
The pinpointing application of this rule adds the new labelled consequence
$A\sqsubseteq B:\axm1\land\axm2$ to $\Amc^{pin}$. Then, Rule 3 becomes pinpointing applicable with
\[
\Bmc_0=\{A\sqsubseteq B,\ A\sqsubseteq A,\ A\sqcap B\sqsubseteq \bot\},
\]
which adds $A\sqsubseteq \bot:\axm1\land\axm2\land \axm4$ to the set of consequences. Then, Rule 7
over the set of premises
\[
\Bmc_0=\{A\sqsubseteq \exists r.A,\ A\sqsubseteq A,\ A \sqsubseteq \forall r.B\},
\] 
yields the new consequence $A\sqsubseteq \exists r.(A\sqcap B):\axm1\land\axm3$. 

Notice that, at this point Rule 6 is not applicable in the classical case over the set of premises 
$\Bmc_0=\{A\sqsubseteq\exists r.(A\sqcap B),A\sqcap B\sqsubseteq\bot\}$ because its (regular) application would
add the consequence $A\sqsubseteq\bot$ that was already derived. However, 
\[
\fm{\Bmc_0}{\Amc^{pin}}=\axm1\land\axm3\land\axm4 \not\models \axm1\land\axm2\land\axm4=\varphi_{A\sqsubseteq\bot}=\fm{\Bmc_1}{\Amc^{pin}};
\]
hence, the rule is in fact pinpointing applicable. The pinpointing application of this Rule 6 substitutes the labelled
consequence $A\sqsubseteq\bot:\axm1\land\axm2\land\axm4$  with the consequence
$A\sqsubseteq\bot:(\axm1\land\axm2\land\axm4)\lor(\axm1\land\axm3\land\axm4)$. The pinpointing extension
will then continue applying rules until a saturated state is reached. This execution is summarized in
Figure~\ref{fig:ppa}.
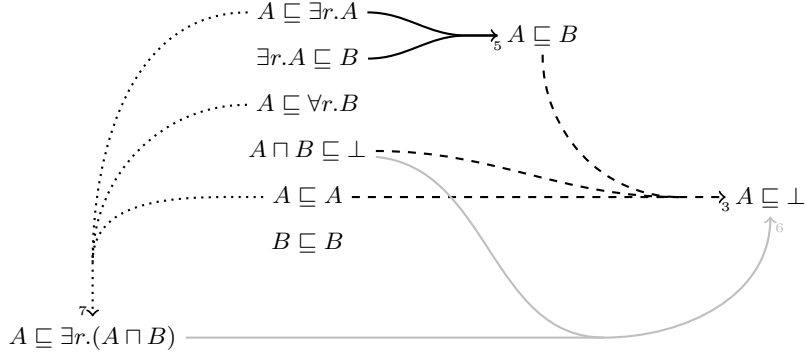
\begin{figure}[tb]
\centering
 \begin{tikzpicture}[thick]
\node (a1) {$A\sqsubseteq\exists r.A$};
\node[below=1mm of a1] (a2) {$\exists r.A\sqsubseteq B$};
\node[below=1mm of a2] (a3) {$A\sqsubseteq\forall r.B$};
\node[below=1mm of a3] (a4) {$A\sqcap B\sqsubseteq \bot$};
\node[below=1mm of a4] (c1) {$A\sqsubseteq A$};
\node[below=1mm of c1] (c2) {$B\sqsubseteq B$};
\node[anchor=west] (c3) at ($(a1)!0.5!(a2)+(2.5,0)$) {$A\sqsubseteq B$};
\draw (a1) to[out=0,in=180] ($(c3)+(-1,0)$) -- (c3);
\draw (a2) to[out=0,in=180] ($(c3)+(-1,0)$) edge[->] node[below=-0.5mm,at end] {\tiny 5} (c3);
\node[anchor=west] (c4) at ($(a4)!1!(c1)+(5.5,0)$) {$A\sqsubseteq \bot$};
\draw[dashed] (a4) to[out=0,in=180] ($(c4)+(-1.2,0)$); 
\draw[dashed] (c3) to[out=-90,in=180] ($(c4)+(-1.2,0)$); 
\draw[dashed] (c1) to[out=0,in=180] ($(c4)+(-1.2,0)$) edge[->] node[below=-0.5mm,at end] {\tiny 3} (c4);
\node[below left=of c2,anchor=east] (c5) {$A\sqsubseteq \exists r.(A\sqcap B)$};
\draw[dotted] (a1) to[out=180,in=90] ($(c5)+(0,1)$); 
\draw[dotted] (a3) to[out=180,in=90] ($(c5)+(0,1)$) edge[->] node[left=-0.5mm,pos=0.9] {\tiny 7} (c5);
\draw[dotted] (c1) to[out=180,in=90] ($(c5)+(0,1)$); 
\node[right=5.5cm of c5] (p) {};
\draw[lightgray] (a4) to[out=-5,in=180] (p);
\draw[lightgray] (c5) to (p.west) edge[->,out=0,in=-90] node[right=-0.5mm,pos=0.95] {\tiny 6} (c4);
\end{tikzpicture}
\caption{Pinpointing application of rules over \Texa in Example~\ref{exa:ppa}.
Arrows point from the premises to the consequences generated by the pinpointing application of the rule denoted in the
subindex.}
\label{fig:ppa}
\end{figure}
At that point, the set of labelled consequences will contain, among others, 
$A\sqsubseteq\bot:(\axm1\land\axm2\land\axm4)\lor(\axm1\land\axm3\land\axm4)$. The label of this 
consequence corresponds to the pinpointing formula that was computed in Example~\ref{exa:cons}.
\end{example}

Notice that if a rule $R$ is applicable to some state $\Amc$, then it is also pinpointing applicable to it. 
This holds because the regular applicability condition requires that at least one consequence $\alpha$ in 
$\res{R}$ should not exist already in the state \Amc, which is equivalent to having the consequence 
$\alpha:\bot\in\Amc^{pin}$. Indeed, we used this fact in the first pinpointing rule applications of 
Example~\ref{exa:ppa}.
If the consequence-based
algorithm is correct, then it follows by definition that for any saturated state $\Amc$ obtained by a sequence of
rule applications from \Omc, $\Omc\models\alpha$ iff $\alpha\in\Amc$. Conversely, as shown next,
every consequence created by a pinpointing rule application is also generated by a regular rule application.
First, we extend the notion of a \valuation-projection to sets of consequences (i.e., states) in the obvious manner: 
$\Amc_\valuation:=\{\alpha\mid \alpha:\varphi_\alpha\in\Amc^{pin},\valuation\models\varphi\}$.
\begin{lemma} 
\label{lem:valuation}
Let $\Amc^{pin},\Bmc^{pin}$ be pinpointing states and let $\valuation$ be a valuation. 
If $\Amc^{pin}\wto^*\Bmc^{pin}$ 
then $\Amc_\valuation\to^*\Bmc_\valuation$. 
\end{lemma}
\begin{proof}
We show that if $\Amc^{pin}\wto_R\Bmc^{pin}$ then $\Amc_\valuation\to_R\Bmc_\valuation$ or 
$\Amc_\valuation=\Bmc_\valuation$, where $R=\Bmc_0\to \Bmc_1$ is a rule. 
If \valuation does not satisfy $\fm{\Bmc_0}{\Amc^{pin}}$ then $\Amc_\valuation=\Bmc_\valuation$
since the labels of the newly added assertions are not satisfied by \valuation, and the 
disjunction with $\fm{\Bmc_0}{\Amc^{pin}}$ does not change the evaluation of 
the modified labels under \valuation. On the other hand,  if 
\valuation satisfies $\fm{\Bmc_0}{\Amc^{pin}}$ then 
$\Bmc_0\subseteq\Amc_\valuation$. If $\Bmc_1\not\subseteq\Amc_\valuation$ 
then  $\Amc_\valuation\to_R\Bmc_\valuation$. Otherwise, again we have 
$\Amc_\valuation=\Bmc_\valuation$. 
\end{proof}
Since all the labels are monotone Boolean formulae, it follows that the valuation $\valuation_\top=\lab{\Omc}$ that
makes every propositional variable true satisfies all labels, and hence for every pinpointing state
$\Amc^{pin}$, $\Amc_{\valuation_\top}=\Amc$. Lemma~\ref{lem:valuation} hence entails that the pinpointing 
extension of the consequence-based algorithm \algorithm does not create new consequences, but only labels 
these consequences.
%
Termination of the pinpointing extension then follows from the termination 
of the consequence-based algorithm and the condition for pinpointing rule application that entails that,
whenever a rule is pinpointing applied, the set of labelled consequences is necessarily modified either by
adding a new consequence, or by modifying the label of at least one existing consequence to a weaker (i.e.,
more general) monotone Boolean formula. Since there are only finitely many monotone Boolean formulas
over $\lab{\Omc}$, every label can be changed finitely many times only. 

It is in fact possible to get a better understanding of the running time of the pinpointing extension of a 
consequence-based algorithm. Suppose that, on input \Omc, the consequence-based algorithm \algorithm
stops after at most $f(\Omc)$ rule applications. Since every rule application must add at least one consequence
to the state, the saturated state reached by this algorithm will have at most $f(\Omc)$ consequences. Consider
now the pinpointing extension of \algorithm. We know, from the previous discussion, that this pinpointing 
extension generates the same set of consequences. Moreover, since there are $2^{|\Omc|}$ possible valuations
over $\lab{\Omc}$, and every pinpointing rule application that does not add a new consequence must 
generalize at least one formula, the labels of each consequence can be modified at most $2^{|\Omc|}$ times.
Overall, this means that the pinpointing extension of \algorithm stops after at most $2^{|\Omc|}f(\Omc)$ rule
applications. We now formalize this result.
\begin{theorem}\label{thm:termination}
If a consequence-based algorithm $\algorithm$ stops after at most $f(\Omc)$ rule applications, 
then $A^{pin}$ stops after at most $2^{|\Omc|}f(\Omc)$ rule
applications.
\end{theorem}


Another important property of the pinpointing extension is that saturatedness of a state is preserved under
projections.
\begin{lemma} 
\label{lem:pin-sat-adapted}
Let  $\Amc^{pin}$ be a pinpointing state and \valuation a valuation. 
If $\Amc^{pin}$ is pinpointing saturated then $\Amc_\valuation$ is saturated. 
\end{lemma}
\begin{proof}
Suppose there is a rule $R$ such that $R$ is applicable to 
$\Amc_\valuation$. This means that $\Bmc_0\subseteq \Amc_\valuation$
and $\Bmc_1\not\subseteq \Amc_\valuation$. We show that $R$ is pinpointing 
applicable to $\Amc^{pin}$. Since  $\Bmc_0\subseteq \Amc_\valuation$, 
\valuation satisfies $\fm{\Bmc_0}{\Amc^{pin}}$.
As $\Bmc_1\not\subseteq \Amc_\valuation$, there is $\alpha\in\Bmc_1$ such that 
either $\consequence\not\in\Amc$ 
or  
$\consequence:\varphi_\consequence\in\Amc^{pin}$
but
\valuation does not satisfy $\varphi_\consequence$.
In the former case, $R$ is clearly pinpointing 
applicable to $\Amc^{pin}$. In the latter, 
$\fm{\Bmc_0}{\Amc^{pin}}\not\models \fm{\Bmc_1}{\Amc^{pin}}$
since   \valuation satisfies
$\fm{\Bmc_0}{\Amc^{pin}}$ but not $\fm{\Bmc_1}{\Amc^{pin}}$.  
\end{proof}

We can now show that the pinpointing extension of a consequence-based algorithm is indeed
a pinpointing algorithm; that is, that when a saturated pinpointing state $\Amc^{pin}$ is reached from rule 
applications starting from $\Omc^{pin}$, then for every $\alpha:\varphi_\alpha\in\Amc^{pin}$, $\varphi_\alpha$
is a pinpointing formula for $\alpha$ w.r.t.\ $\Omc^{pin}$.
\begin{theorem}[Correctness of Pinpointing]\label{thm:correct-pinpointing}
Let $\cproperty$ be a c-property on axiomatized inputs for $\inputs$ and $\axiomsset$.
Given a correct consequence-based algorithm $\algorithm$ for $\cproperty$,  for every 
axiomatized input $\Gamma=(\consequence,\axioms)\in\cproperty$, where \axioms is normalized, then
\begin{quote}
 if $\axioms^{pin}\wto^* \Amc^{pin}$, 
$\consequence:\pinpointing{\consequence}\in \Amc^{pin}$, and 
$\Amc^{pin}$ is pinpointing saturated,
then \pinpointing{\consequence} 
is a pinpointing formula for \cproperty and $\Gamma$. 
\end{quote}
\end{theorem}
\begin{proof}
We want to show that  \pinpointing{\consequence} is a pinpointing formula for 
\cproperty and $(\consequence,\axioms)$. That is, for every valuation $\valuation\subseteq \lab{\axioms}$:
$(\consequence,\axioms_\valuation)\in \cproperty$ iff \valuation satisfies $\pinpointing{\consequence}$.  

Assume that $(\consequence,\Omc_\valuation)\in \cproperty$, i.e., $\Omc_\valuation\models\consequence$, and
let $\Amc_0=\Omc_\valuation$. Since \algorithm terminates on every input, there is a saturated state
$\Bmc$ such that $\Amc_0\to^*\Bmc$. Completeness of \algorithm then implies that $\alpha\in\Bmc$.
By assumption, $\Amc_0^{pin}\wto^*\Amc^{pin}$ and $\Amc^{pin}$ is pinpointing
saturated. By Lemma~\ref{lem:valuation} it follows that $\axioms_\valuation\wto^* \Amc_\valuation$, and
by Lemma~\ref{lem:pin-sat-adapted}, $\Amc_\valuation$ is saturated. Hence,
since \algorithm is correct, $\Amc_\valuation=\Bmc$. This implies that 
$\valuation\models\pinpointing{\consequence}$ because $\alpha\in\Amc_\valuation$.

Conversely, suppose that \valuation satisfies $\pinpointing{\consequence}$.  
By assumption, $\consequence : \pinpointing{\consequence} \in \Amc^{pin}$,
$\axioms^{pin}\wto^* \Amc^{pin}$, 
and $\Amc^{pin}$ is saturated. By Lemma~\ref{lem:valuation}, 
$\axioms_\valuation\to^* \Amc_\valuation$. 
Since \valuation satisfies $\pinpointing{\consequence}$, 
$\alpha\in \Amc_\valuation$. Then, by soundness of $A$, $\axioms_\valuation\models\alpha$. 
\end{proof}
As it was the case for classical consequence-based algorithms, their pinpointing extensions can apply the
rules in any desired order. The notion of correctness of consequence-based algorithms guarantees that a 
saturated state will always be found, and the
result will be the same, regardless of the order in which the rules are applied. We have previously seen that 
termination transfers also the pinpointing extensions. Theorem~\ref{thm:correct-pinpointing} also shows that 
the formula associated to the consequences derived is always equivalent.
\begin{corollary}
Let $\Amc^{pin},\Bmc^{pin}$ two pinpointing saturated states, \axioms an ontology, and \consequence a 
consequence such that $\alpha:\varphi_\alpha\in\Amc^{pin}$ and $\alpha:\psi_\alpha\in\Bmc^{pin}$. 
If $\axioms^{pin}\wto^*\Amc^{pin}$ and $\axioms^{pin}\wto^*\Bmc^{pin}$, then 
$\varphi_\alpha\equiv\psi_\alpha$.
\end{corollary}

To finalize this section, we consider again our running example of deciding subsumption in \ALC 
described in 
Section~\ref{sec:consequence-based-alg}.
It terminates after an exponential number of rule applications on
the size of the input TBox \Tmc. Notice that every pinpointing rule application requires an entailment test between
two monotone Boolean formulas, which can be decided in non-deterministic polynomial time on $|\Omc|$.
Thus, it follows from Theorem~\ref{thm:termination} 
that the pinpointing extension of the consequence-based algorithm for \ALC runs in exponential time.

%
\begin{corollary}
Let \Tmc be an \ALC TBox, and $C,D$ two \ALC concepts. A pinpointing formula for $C\sqsubseteq D$ w.r.t.\
\Tmc is computable in exponential time.
\end{corollary}
 
\section{Dealing with Normalization}

Throughout the last two sections, we have disregarded the first phase of the consequence-based algorithms
in which the axioms in the input ontology are transformed into a suitable normal form. In a nutshell, the
normalization phase takes every axiom in the ontology and substitutes it by a set of simpler axioms that are,
in combination, equivalent to the original one w.r.t.\ the set of derivable consequences. For example, in \ALC the
axiom $A\sqsubseteq B\sqcap C$ is \emph{not} in normal form. During the normalization phase, it would then
be substituted by the two axioms $A\sqsubseteq B$, $A\sqsubseteq C$, which in combination provide the exact
same constraints as the original axiom.

Obviously, in the context of pinpointing, we are interested in finding the set of \emph{original} axioms that
cause the consequence of interest, and not those in normal form; in fact, normalization is an internal process
of the algorithm, and the user should be agnostic to the internal structures used. Hence, we need to find a way
to track the original axioms. 

To solve this, we slightly modify the initialization of the pinpointing extension. Recall from the previous section that,
if the input ontology is already in normal form, then we initialize the algorithm with the state that contains exactly
that ontology, where every axiom is labelled with the unique propositional variable that represents it. If the 
ontology is not originally in normal form, then it is first normalized. In this case, we set as the initial state the
newly normalized ontology, but every axiom is labelled with the disjunction of the variables   
representing the axioms that 
generated it. The following example explains this idea.

\begin{example}
Consider a variant $\Texa'$ of the \ALC TBox from Example~\ref{exa:cons} that is now formed by the three axioms
\begin{align*}
A\sqsubseteq\exists r.A\sqcap\forall r.B: {} & \axp1, &
\exists r.A\sqsubseteq B: {} & \axp2, \\ &&
A\sqcap B\sqsubseteq\bot: {} & \axp4.
\end{align*}
Obviously, the first axiom \axp1 is not in normal form, but can be normalized by substituting it with the two
axioms $A\sqsubseteq\exists r.A$, $A\sqsubseteq\forall r.B$. Thus, the normalization step yields the same
TBox \Texa from Example~\ref{exa:cons}. However, instead of using different propositional variables to label
these two axioms, they just inherit the label from the axiom that generated them; in this case \axp1. Thus,
the pinpointing algorithm is initialized with
\[
\Amc^{pin}=\{
 A\sqsubseteq\exists r.A:\axp1,
 \exists r.A\sqsubseteq B:\axp2,
 A\sqsubseteq\forall r.B:\axp1,
 A\sqcap B\sqsubseteq\bot:\axp4
\}.
\]
Following the same process as in Example~\ref{exa:ppa}, we see that we can derive the consequence
$A\sqsubseteq\bot:\axp1\land\axp4$. Hence $\axp1\land\axp4$ is a pinpointing formula for $A\sqsubseteq\bot$
w.r.t.\ $\Texa'$. It can be easily verified that this is in fact the case.
\end{example}

Thus, the normalization phase does not affect the correctness, nor the complexity of the pinpointing extension
of a consequence-based algorithm.

\section{Conclusions}

We presented a general framework to extend consequence-based algorithms 
with axiom pinpointing. These algorithms often enjoy optimal upper bound 
complexity and can be efficiently implemented in practice. 
Our focus in this paper and use case is for the prototypical 
ontology language \ALC. 
We emphasize that this is only one of many consequence-based algorithms available. 
The completion-based algorithm for \elplus~\cite{baader-ijcai05} is obtained by restricting the assertions 
to be of the form
$A\sqsubseteq B$ and $A\sqsubseteq\exists r.B$ with $A,B\in N_C\cup\{\top\}$ and $r\in N_R$, 
and adding one rule to handle role constructors.
Other examples of consequence-based methods include LTUR approach for Horn clauses~\cite{minoux-ipl88}, 
and methods for more expressive 
and Horn DLs~\cite{Kaza09,horrocks-kr16,kazakov-jar14}. 

Understanding the axiomatic causes for a consequence, and in particular the pinpointing formula, has importance
beyond MinA enumeration. For example, the pinpointing formula also encodes all the ways to \emph{repair}
an ontology~\cite{arif-ki15}. Depending on the application in hand, a simpler version of the formula can be 
computed, potentially more efficiently. This idea has already been employed to find good approximations
for MinAs~\cite{BaaSun-KRMED-08} and lean kernels~\cite{PMIM17} efficiently.

As future work, it would be interesting to investigate how 
algorithms for query answering in an ontology-based data access setting 
can be extended with the pinpointing technique. The pinpointing formula 
in this case could also be seen as a provenance polynomial, as introduced 
by~Green et. al~\cite{Green07-provenance-seminal}, 
in database theory. Another direction is to investigate axiom pinpointing 
in decision procedures for non-monotonic reasoning, where one would also 
expect the presence of negations in the pinpointing formula.

%

\bibliographystyle{plain}
\bibliography{references,local}

\end{document}